\documentclass[12pt,reqno]{amsart}
\usepackage{dsfont, amssymb,amsmath,amscd,latexsym, amsthm, amsxtra,amsfonts,enumerate}
\usepackage[all]{xy}
\usepackage[active]{srcltx}
\usepackage{mathrsfs}
\usepackage[all]{xy}
\usepackage[active]{srcltx}
\usepackage{graphicx}
\usepackage{bm}
\usepackage{bbm}
\usepackage{graphicx}
\usepackage{tabularx}
\usepackage{caption}
\usepackage{color}
\usepackage{algorithmicx,algorithm}
\usepackage{longtable}
\usepackage[round]{natbib}
\usepackage{ulem}
\usepackage{verbatim}
\usepackage{appendix}
\usepackage[pdfstartview=FitH, bookmarksnumbered=true,bookmarksopen=true, colorlinks=true, pdfborder=001, citecolor=blue, linkcolor=blue,urlcolor=blue]{hyperref}
\usepackage{CJK}
\newtheorem{theorem}{Theorem}[section]
\newtheorem{lemma}[theorem]{Lemma}

\renewcommand{\theequation}{\arabic{section}.\arabic{equation}}

\newcommand{\E}{\mathrm{E}}
\renewcommand{\theequation}{\arabic{section}.\arabic {equation}}

\begin{document}
\makeatletter
\def\@setauthors{%
\begingroup
\def\thanks{\protect\thanks@warning}%
\trivlist \centering\footnotesize \@topsep30\p@\relax
\advance\@topsep by -\baselineskip
\item\relax
\author@andify\authors
\def\\{\protect\linebreak}%
{\authors}%
\ifx\@empty\contribs \else ,\penalty-3 \space \@setcontribs
\@closetoccontribs \fi
\endtrivlist
\endgroup } \makeatother
 \baselineskip 19pt
\title[{{\tiny Optimal  retirement decision  and consumption  }}]
 {{\tiny Optimal  Retirement Time and  Consumption with the Variation in Habitual Persistence }}
 \vskip 10pt\noindent
\author[{Lin He, Zongxia Liang, Yilun Song, Qi Ye}]
{\tiny {\tiny  Lin He$^{a,\dag}$, Zongxia Liang$^{b,\ddag}$, Yilun Song$^{b,\S}$,
Qi Ye$^{b,*}$ }
 \vskip 10pt\noindent
{\tiny ${}^a$School of Finance, Renmin University of China, Beijing
100872, China \vskip 10pt\noindent\tiny ${}^b$Department of
Mathematical Sciences, Tsinghua University, Beijing 100084, China }
  \footnote{\\
$ \dag$ email: helin@ruc.edu.cn\\
 $ \ddag$ email: liangzongxia@mail.tsinghua.edu.cn\\
 $\S$   email: songyl18@mails.tsinghua.edu.cn\\
 $*$ Corresponding author, email:   yeq19@mails.tsinghua.edu.cn  }}
\maketitle
 \noindent
%
\begin{abstract}
 In this paper,we study the individual's optimal retirement time and optimal consumption under habitual persistence. Because the individual feels equally satisfied with a lower habitual level and is more reluctant to change the habitual level after retirement, we assume that both the level and the sensitivity of the habitual consumption decline at the time of retirement.  We establish the concise form of the habitual evolutions, and obtain the optimal retirement time and consumption policy  based on martingale and duality methods. The optimal consumption experiences a sharp decline at retirement, but the excess consumption raises because of the reduced sensitivity of the habitual level. This result contributes to explain the ``retirement consumption puzzle". Particularly, the optimal retirement and consumption   policies are balanced between the wealth effect and the habitual effect. Larger wealth increases consumption, and larger growth inertia (sensitivity) of the habitual level decreases  consumption and brings forward the retirement time.

\vskip 10 pt  \noindent
 Submission Classification: IB13, IE13, IM12.
\end{abstract}
\vskip 15pt
\setcounter{equation}{0}
\section{\bf Introduction}
In empirical evidence, we observe a sharp decline in the consumption
at retirement, which can not be explained by the classical
\cite{Merton1969, Merton1983} model of lifetime optimal consumption.
Literatures explain this ``retirement consumption puzzle" as the
results of more household productions, less work-related expenses,
restricted retirement incomes and more medical expenses, etc (cf.
\cite{BBT1998}, \cite{BSW2001}, \cite{HR2003}, \cite{H2003},
 \cite{S2005}, \cite{B2008}).
However, the explanations are still a bit controversial and need more
theoretical evidence. In this paper, we try to establish the
individual's optimal  retirement decision  and consumption problem under  habitual
persistence to clear the controversy in some extent, and establish some practical settings on the habitual
parameters. The results are explanatory to the empirical evidence.
\vskip 5pt The habit persistence was originally studied by
\cite{Pollak70} and \cite{RH73}.The habitual level is the weighted
average of the past consumption and only the excess consumption
produces utility. Basically, the optimization problems are divided
into two categories according to whether the consumption can be
lower than the  habitual level. \cite{DZ91}, \cite{Chapman98},
\cite{Yu15} and \cite{GLY2021} study the optimal consumption problem
under the additive habit framework. And \cite{Shrikhande97} and
\cite{DK03} study the problem under the non-addictive framework.
Under this framework, it is usually assumed that there exist the
minimal consumption constraints \  (cf. \cite{LLS2018}  and
\cite{ABY2020}).  We assume that the actual consumption can be lower
than the habitual level and the downward deviation is bounded for
the practical considerations. In this circumstance, the CRRA
(Constant Relative Risk Averse) utility is not well defined.
Inspired by \cite{Curatola2017} and \cite{BLN2020}, which study the
optimal consumption problem under loss aversion and habitual
reference, we explore the S-shaped utility in this paper. \vskip 5pt
The individual's life cycle is divided into two periods. He/She
earns wages and makes contributions to the social endowment
insurance before retirement, and receives benefits after retirement.
Hence the individual's wealth process is naturally described by
two-stage stochastic differential equations, as in \cite{CHX2018}.
For the retirement period, the individual's income usually declines
due to the low income replacement rate. However, the individual has
more leisure time in this period. As such, several ways are
implemented to depict the utility of leisure. In \cite{BBT1998} and
\cite{KPS2004}, a weight parameter larger than one is set to the
retirement consumption. This leads to a jump in the consumption at
retirement and contributes to explain the ``retirement consumption
puzzle". Then, \cite{CHX2018} assume a leisure weight related to the
retirement time and obtain interesting findings. \vskip 5pt It is
still very challenging to describe the utility of leisure in the
framework of habitual persistence. According to the statistical
results of the survey, we observe that the individual feels equally
satisfied with a lower actual consumption level after retirement.
This may be explained as that the individual has more leisure time
and makes more household productions. We model this effect as the
shrinking of the habitual level at the time of retirement. Under
this assumption, the same actual consumption leads to more excess
consumption and greater utility after retirement. Besides, we find
that the individuals are reluctant to change their habitual levels
after retirement. We observe that one time large (low) consumption has
limited impacts on the habitual consumption level, which is caused by the income constraint and hard to change
habit. As such, we assume that the sensitivity of the habitual
consumption declines after retirement. Under the relatively abundant
wealth scenario in practice, we expect a spiral increase of the
habitual consumption level. Because the inertia of the habitual growth
rate declines after retirement, the same actual consumption leads to
greater utility. \vskip 5pt The individual dynamically controls the
optimal asset allocation and consumption policies to achieve the
objective, i.e., maximizing the overall utility of the excess
consumption. One crucial parameter is the retirement time. Thus, the
problem is transformed into a two-stage optimization problem. In the
former stage, the individual's utility is maximized under the
assumption that the retirement time is given. In the latter stage,
the optimal retirement time is regarded as a pre-commitment policy,
which is consistent with the settings in \cite{HL2009},
\cite{HP2011} and \cite{CHX2018}. Literatures like
\cite{CS2006},\\\cite{CSS2008}, \cite{DL2010},\cite{KY2018} and
\cite{GLY2021} treat the optimal retirement time as a simultaneously
controlled policy and study the optimal stopping time problem with
free boundary. However, we believe that the pre-commitment
assumption is more practical. In this circumstance, the retirement
decisions are made relatively ahead of time, and do not depend on
the follow-up situations. This can be verified by the consistency
between the expected and the actual retirement time in the survey.
\vskip 5pt In order to establish the optimums of the individual's
optimization problem, we introduce the dirac function to represent
the dynamics of the habitual level concisely in a unified
expression. Using martingale and duality methods, as in
\cite{KLSX1991}, \cite{HLLM2020} and \cite{HLY2020}, we establish
the semi-analytical solutions of the problem. Because of the long
time horizon of the life cycle optimization, the expectation of the
optimal consumption may explode after a long time in the numerical
simulations. As such, we originally introduce the certainty
equivalence of the actual and the habitual consumption processes timely.
In addition, we establish an innovative analytical method to study
the impact of the exogenous parameters on the optimal retirement
time quantitatively. The theoretical and numerical results show that
the optimal consumption is affected by both of the wealth effect and
the habitual effect. Particularly, the optimal consumption decision
is balanced between the pressure of high habitual consumption level
and the time preference. The declines of the level and the
sensitivity of the habitual level reduce the pressure and weaken the
habitual effect after retirement. The excess consumption rises at
the time of retirement accordingly. However, the rise of the excess
consumption cannot offset the sharp decline in the habitual
consumption. Thus, we observe a decline in the actual consumption at
the time of retirement, which is consistent with the empirical
evidence of the ``retirement consumption puzzle". Furthermore,
larger wealth increases consumption, and larger wage (benefit)
postpones (brings forward) the retirement time. Interestingly,
larger growth inertia (sensitivity) of the habitual level decreases
consumption and brings forward the retirement time. In this
circumstance, the habitual effect is the dominance. Early retirement
is required to prevent the habitual level from raising too high and
reducing overall utility. \vskip 5pt The main contributions of this
paper are threefold: First, we establish the non-addictive optimal
retirement decision  and consumption problem under the framework of habitual persistence.
Particularly, we model the utility of the retirement leisure as the
declines in the level and the sensitivity of the habitual
consumption after retirement. Second, we first introduce the dirac function
to establish the unified expression of the habitual level under the
settings of variational habitual persistence characters. Through
this simple expression and
 using martingale and duality methods, we
establish the semi-analytical optimums of the stochastic control
problem. In addition, we quantitatively analyze the relationship
between the exogenous parameters and the optimal retirement time.
The last, the numerical results confirm the evidence that there
exists a sharp decline in the consumption at retirement. It is the comprehensive result
of the drop in the habitual consumption and the rise in the excess consumption. Besides, we identify the wealth effect and the
habitual effect, which have major impacts on the optimal consumption
and the optimal retirement time. \vskip 5pt The remainder of this
paper is organized as follows: Section 2 formulates the
non-addictive optimal consumption problem with the variation in the
habitual level and the habitual sensitivity. In Section 3, we
establish the optimal control policies and the value function
semi-analytically based on dirac function, martingale and duality
methods. Section 4 shows the optimal consumption process
numerically, which contributes to explain the ``retirement
consumption puzzle". In addition, we study the impacts of the
parameters on the optimal consumption and the optimal retirement
time in this section. The last section concludes the paper. \vskip
25pt \setcounter{equation}{0}
\section{\bf Problem formulation}
\vskip 5pt
\subsection{\bf Wealth process with the participation of social endowment insurance}

We consider the wealth process of an individual who participates in
the social endowment insurance. Before retirement, the individual
receives wage as the labor income and contributes part of the income
as the premium of the social insurance. After retirement, the
individual receives benefit from the social insurance. Meanwhile,
investment is allowed both before and after retirement. The
individual dynamically chooses the asset allocation and consumption
policies, as well as the retirement time to maximize the overall
utility of the consumption. \vskip 5pt First, let $B=\{B_t, \ t\geq
0\}$ be a Brownian motion  on  complete filtered  probability $
(\Omega,\mathcal{F},\mathbb{F},\mathbb{P}) $, the filtration  $
\mathbb{F}\!= \!\{\mathcal{F}_t: \ t\geq 0\} $ satisfies the usual
conditions and is generated by the Brownian motion $B$, i.e., $
\mathcal{F}_t=\sigma\{B_{u}, \ 0\leq u\leq t\}$,  $t\geq 0$. For
simplicity, we assume that the financial market consists of one
risk-free asset and one risky asset. \vskip 5pt The price of the
risk-free  asset $S_0=\{S_0(t), \ t\geq 0\}$  is  given by
\begin{align*}
{\frac{dS_0(t)}{S_0(t)}} = rdt, \  S_0(0)=s_0,
\end{align*}
and the price of the risky asset  $ S_1=\{S_1(t), \ t\geq 0\}$
follows the stochastic differential equation (abbr.  SDE):
\begin{align*}
{\frac{dS_1(t)}{S_1(t)}} = \mu dt+\sigma dB_t,\  S_1(0)=s_1,
\end{align*}
where $r$ is the risk-free interest rate,  $\mu$ and $\sigma$ are the
expected return and the volatility of the risky asset, respectively.
$s_0$ and $s_1$ are positive constants.
\vskip 5pt
The wage process $W =\{W(t), \  t\geq 0\}$ satisfies the following SDE:
\begin{align*}
\frac{dW(t)}{W(t)}=\alpha(t) dt+\beta(t) dB_t,\  W(0)=W_0,
\end{align*}
where $\alpha(t)$ and $\beta(t)$ are the expected growth rate and
the volatility of the wage at time $t$. \vskip 5pt Next, we
establish the settings of the social insurance rules and the
retirement time. The individual chooses his/her optimal retirement
time $\tau$, which lies within $[\tau_{min},\tau_{max}]$.
$\tau_{min}$ and $\tau_{max}$ are minimal and maximal retirement
time required by the government. In addition, we assume that the
optimal retirement time is a pre-commitment policy, that is, the
individual chooses the optimal time to retire at the initial time
rather than choosing it by constantly observing the follow-up
situations. The settings are consistent with the ones in
\cite{HP2011} and \cite{CHX2018}. \vskip 5pt Before retirement, the
individual mandatorily participates in the social endowment
insurance, and has the obligation to contribute $k$ proportion of
the wage as the insurance premium. After retirement, the individual
has the right to receive the benefit from the social insurance. The
amount of the benefit depends on the retirement time. Particularly,
the individual receives the amount of $g(\tau)De^{\xi t}$ as the
benefit at time $t$, where $D$ is the benefit at time $0$, $\xi$ is the growth
rate of the benefit, which is designed to maintain the purchasing
power, and  $g(\tau)$ is a penalty variable for early retirement. We
assume  $g(\tau)=e^{-\zeta(\tau_{st}-\tau)^+}$, where $\tau_{st}$ is
the statutory retirement time and $\zeta>0$ is the elastic
parameter. As such, $g(\tau)$ is increasing with respect to $\tau$
before $\tau_{st}$ and invariant after $\tau_{st}$. For early
retirement $(\tau<\tau_{st})$, the individual receives discounted
benefit and is encouraged to work to the statutory retirement age.
For the wage process, we assume $\alpha(t)=\alpha 1_{\{t\leq
\tau_{min}\}}$ and $\beta(t)=\beta 1_{\{t\leq \tau_{min}\}}$, where
$\alpha$ and $\beta$ are positive constants. It is a realistic
setting that the wage stops rising after a certain time. \vskip 5pt
The last, we establish the individual's wealth process. Assume that
$T$ is the maximal survival time. The dynamics of the wealth $X=\{
X_t , \  0\leq t\leq T\}$ is determined by three control variables:
the wealth allocated to the risky asset $\pi=\{\pi_t, \ 0\leq t\leq
T \}$, the consumption level $C=\{ C_t , 0\leq t\leq T \}$ and the
retirement time $\tau$. It is natural to assume that  $C:[0,T]\times
\Omega \to [0,\infty)$ and $\pi:[0,T]\times \Omega \to [0,\infty)$
are $\mathbb{F}$-progressively measurable and satisfy the
integrability condition $\int_0^T (C_t+\pi_t^2) dt<\infty$ almost
surely. Besides, $\tau$ is a pre-commitment control variable, as
such, we treat it as a given parameter in the first step of
optimization. \vskip 5pt The individual's wealth process satisfies
the following SDEs: \\ When $0\leq t\leq \tau$,
\begin{align*}
dX_t=rX_t + (\mu-r)\pi_t dt+(1-k)W_t dt+\sigma \pi_t dB_t -C_t
dt,
\end{align*}
and when $\tau\leq t\leq T$,
\begin{align*}
dX_t=rX_t + (\mu-r)\pi_t dt+g(\tau)D e^{\xi t} dt+\sigma \pi_t
dB_t -C_t dt.
\end{align*}
\vskip 5pt In the next subsection, we will establish the
individual's optimization objective and the admissible domain of the
control variables. \vskip 10pt
\subsection{Optimization objective with the variation in habitual consumption}
In this subsection, we establish the S-shaped utility function of
the individual and the variations in the level and the sensitivity
of the habitual consumption after retirement, which are better
depictions of the reality. \vskip 5pt First, we set up the habitual
consumption level of the individual. The habitual behavior was
originally studied by \cite{Pollak70} and \cite{RH73}, which
establish the criterion of evaluating consumption, and the habitual
level is measured by the weighted average of the past consumption.
Composing the S-shaped utility, only the consumption
exceeding the habitual level produces positive utility. The habitual
level $h=\{h_t, \ 0\leq t\leq T\}$ is defined by
\begin{align*}
dh_t=\left [\psi(t) C_t-\widetilde{\eta(t)} h_t\right ]dt,\  t\neq \tau,
\end{align*}
where $\psi(t)$ and $\widetilde{\eta(t)}>0$ are nonnegative and
bounded parameters varying with respect to time $t$.
$\psi(t)=\widetilde{\eta(t)}$ usually holds. Under this assumption,
the habitual consumption level is the arithmetic
average of the past consumption, as discussed in \cite{BLN2020}.
Particularly, when the actual consumption equals to the habitual
consumption, the habitual level remains unchanged. \vskip 5pt Next,
we establish the two important variations in the habitual
consumption after retirement. Briefly, both the level and the
sensitivity of the habitual consumption decline at the time of
retirement. The modifications of  $\{\psi(t), \  0\leq t\leq T\}$
and $\{\widetilde{\eta(t)},\  0\leq t\leq T\}$ are based on two
observations. The former observation is that the individual has more
leisure time and is able to make more household productions. As
such, the individual feels equally satisfied when the consumption
level declines. Because only the difference between the actual
consumption and the habitual consumption produces utility, this
observation could be modeled as the shrinking of habitual
consumption level after retirement. The latter observation is that
the retired individual's habitual level is sluggish to change over
time. The consumption after retirement is influenced by the wealth
constraints and the daily routines. It is unusual that one sudden
extremely high (low) consumption could change the habitual
consumption level a lot. Thus, the sensitivity of the variation in
the habitual consumption level also decreases after retirement.
\vskip 5pt Based on the above discussions, we assume $h_\tau=l
h_{\tau^-}$, $0<l <1$, which refers to the decline of the habitual
consumption level after retirement. Furthermore, we assume
$\psi(t)=\psi$, $\widetilde{\eta(t)}=\eta$ when $0\leq t<\tau$;
$\psi(t)=m\psi$, $\widetilde{\eta(t)}=m\eta$ when $\tau\leq t\leq
T$, where $\psi$ and $\eta$ are constants. And $0<m<1$ refers to the
decline of the sensitivity of the habitual consumption level. As
such, the habitual consumption level $h=\{h_t, \ 0\leq t\leq T\}$ is
given by
\begin{eqnarray}\label{habit}
h_t=\left\{ \begin{array}{lll}
&h_0 \cdot e^{-\eta t}+\psi \cdot \int_0^t e^{-\eta(t-s)}C_s ds, &0\leq t<\tau,\\
&l\cdot h_0\cdot e^{-m\eta(t-\tau)-\eta\tau}+l\cdot \psi\cdot \int_0^\tau
e^{-m\eta(t-\tau)-\eta(\tau-s)}C_s ds\\
&+m\cdot\psi\cdot \int_\tau^t
e^{-m\eta(t-s)}C_s ds, &\tau \leq t\leq T.
\end{array}\right.
\end{eqnarray}
In order to obtain the concise representation of $h_t$, we introduce
the dirac function $\delta(x)$, which is a generalized function
defined  by
\begin{align*}
\delta(x)=\left\{ \begin{array}{ll}
\infty, &x=0,\\
0, &x\neq 0,
\end{array}\right.
\end{align*}
and
\begin{align*}
\int_{-\infty}^{+\infty} \delta(x) dx=1.
\end{align*}
By substituting $\eta(s)$ with $\widetilde{\eta(s)} - \ln(l)
\delta(s-\tau)$ in \eqref{habit}, $h$ is transformed into
\begin{align}\label{habitnew2}
h_t=h_0 e^{-\int_0^t \eta(s)ds}+ \int_0^t e^{-\int_s^t
\eta(u)du}\psi(s)C_s ds.
\end{align}
As such,  $h$ has the following concise form:
\begin{align}\label{habitnew}
dh_t=\left [\psi(t) C_t-\eta(t)h_t\right ]dt,\  ~0\leq t\leq T.
\end{align}
In addition,   $\delta(\cdot)$  in \eqref{habitnew2} can also be  an indicator to distinguish the
pre-retirement and post-retirement scenarios.
\vskip 5pt The last,
we establish the S-shaped utility function of the individual. As
discussed above, only the difference between the actual consumption
and the habitual consumption produces utility. Besides, we assume
that  the actual consumption below the habitual level is allowed. As
such, we need a well defined utility function to express positive
and negative utilities synthetically, and we naturally choose the
S-shaped utility function. \vskip 5pt The S-shaped utility function
is represented as the compositions of two CRRA utilities:
\begin{align}\label{sutility}
u(C_s-h_s)=\frac{(C_s-h_s)^{1-\gamma}}{1-\gamma}1_{\{C_s\geq
h_s\}}+(-\kappa)\frac{(h_s-C_s)^{1-\gamma}}{1-\gamma}1_{\{C_s<
h_s\}},\  \forall
s\in [0, T],
\end{align}
where $\gamma$ refers to the relative risk aversion of the
individual, and $\kappa$ is the loss aversion parameter. The second
term in \eqref{sutility} defines the negative utility when the
individual's actual consumption is below the habitual level. In this
circumstance, the individual suffers from the gap of inadequate
consumption, as such, it produces negative utility. $\kappa>1$
represents the psychological phenomenon that the suffering of loss
is greater than the happiness of equal gain. \vskip 5pt Therefore,
we assume that the retirement time $\tau$ is given in the first
step. The individual's objective is to maximize the overall utility
of the consumption:
\begin{align}\label{objective}
V(\tau)=\max_{\pi, C} \E \left\{\int_0^T e^{-\rho s} u(C_s,h_s) ds\right\},
\end{align}
where $\rho$ refers to the rate of time preference along with the
mortality rate. Actually, the integral in \eqref{objective} is
divided into two parts. Before retirement, the individual earns wage
and contributes to the social endowment insurance. After retirement,
the individual receives benefit. At the same time, the
level and the sensitivity of the habitual consumption both decline
in this period. In the second step, we try to establish
$\arg\max\limits_{\tau}\{ V(\tau)\}$ satisfying
\begin{align}\label{objective2}
V=\max_{\tau}\left\{V(\tau)\right\},
\end{align}
which is the solution of the static optimization problem. \vskip 5pt
Furthermore, we study the restrictions on the admissible domain  of
the control variables. Although we set up the non-addictive
consumption, the actual consumption should be above some threshold
to maintain the minimal standard of living. We assume that there
exists $L\geq 0$ such that $h_s-C_s\leq L$ holds for $ s\in [0,
T]$. Particularly, when $L=0$,  the actual consumption below the
habitual consumption is not allowed and it depicts the addictive
consumption model. In addition, the individual may initiate some
loans to consume the labor capital in advance. And he/she is
required to repay the loans at the maximal survival time. As such,
the triple control $\mathfrak{a}\triangleq(\tau,\pi,C)$ is called
admissible if the individual's wealth $X_T^\mathfrak{a}$, under the
$\mathfrak{a}$, remains nonnegative at time $T$, i.e.,
\begin{align*}
X_T^{\mathfrak{a}}\geq 0
\end{align*}
almost surely. We denote the family of admissible triple controls
 $\mathfrak{a}\triangleq(\tau,\pi,C)$ by $\mathcal{A}$.
\vskip 5pt Our goal is to solve the individual's optimization
problem (\ref{objective})-\eqref{objective2} within the admissible
domain $\mathcal{A}$. The optimal asset allocation and consumption
policies will be established correspondingly. In addition, the
optimal retirement time will be obtained by the pre-commitment
optimization after the value function $V(\tau)$ being
established. \vskip 15pt \setcounter{equation}{0}
\section{\bf Solution of the stochastic optimization problem}
In this section, we assume that the retirement time is given in the
first step of optimization. Using martingale and duality methods, we
derive the semi-analytical value function of the
individual and establish the optimal asset allocation and
consumption policies correspondingly. In the second step of
optimization, we treat the optimal retirement time as a
pre-commitment policy and establish the optimal solution by
numerical methods. Interestingly, we can quantitatively analyze the
relationship between the optimal retirement time and the exogenous
parameters by introducing an innovative analytical method. \vskip
5pt First, we solve the individual's stochastic optimization problem
with the given retirement time $\tau$. The state price density
process $H=\{H_t,\  0\leq t\leq T\}$ is  as follows:
\begin{align*}
H_t=\exp \{ -rt-\frac{1}{2}\theta^2 t-\theta B_t\},\ t\in [0, T],
\end{align*}
where  $\theta=\frac{\mu-r}{\sigma}$ is the market price  of risk.
Using It\^o's lemma,  we have
\begin{eqnarray}\label{wealth}
H_t X_t&=&\int_0^t\left[ (1-k)H_sW_s1_{\{s\leq \tau\}}+ g(\tau)De^{\xi t}
H_s1_{\{s>\tau\}}-C_sH_s \right ]ds\nonumber\\
&&+\int_0^t\left[ \sigma \pi_s H_s-\theta X_sH_s\right] dB_s, \ t\in [0, T].
\end{eqnarray}
For any $(\pi,C)\in\mathcal{A}_\tau \triangleq \{( \pi,C) : (\tau,\pi,C)\in \mathcal{A}       \}$ for given $ \tau $, $X_t\geq \underline{X_t}$ is valid,
and  $\underline{X_t}$ is the lower bound of $X_t$. The details to
obtain the explicit form of $\underline{X_t}$ is given in Appendix \ref{A.1}.  Thus, $\{\int_0^t \left [\sigma \pi_s H_s-\theta X_sH_s\right ] dB_s, \ 0\leq t\leq T\}$ is a
supermartingale. Using the restriction that $X_T\geq 0$, we know that the consumption level satisfies the following
restrictions:
\begin{align}\label{termone}
\E\left[\int_0^T C_sH_s ds\right]\leq \E\left \{\int_0^T\left[ (1-k)H_sW_s1_{\{s\leq \tau\}}+
g(\tau)De^{\xi t} H_s1_{\{s>\tau\}}\right ]ds\right\}.
\end{align}
However, the habitual consumption process is introduced in the
utility function and only the difference between the actual
consumption and the habitual consumption produces utility. As such,
$c_s\triangleq C_s-h_s$ is the actual control variable, and the state price
density process should be adjusted accordingly to solve the
optimization problem. \vskip 5pt The adjusted state price density
process $\Gamma=\{  \Gamma_t, \ 0\leq t\leq T\}$ is defined by
\begin{align*}
\Gamma_t\triangleq H_t+\psi(t) \E_t\left \{\int_t^T e^{-\int_t^s
\eta(u)-\psi(u) du}H_sds\right \}, \quad t\in[0,T].
\end{align*}
And $\Gamma $  is the solution of the following recursive linear
stochastic equation (cf. \cite{DZ92}):
\begin{align*}
\Gamma_t=H_t+\psi(t) \E_t\left\{\int_t^T e^{-\int_t^s \eta(u)du} \Gamma_s
ds\right \}, \quad t\in[0,T],
\end{align*}
where $\E_t(\cdot)\triangleq\E\left[\ \cdot\ |\mathcal{F}_t\right ]$.
And its exact value is derived in Appendix \ref{A.2}. Particularly,
when $\psi=0$, $\Gamma_t$ degenerates to $H_t$. \vskip 5pt
Substituting \eqref{habitnew2}, we have
\begin{align}\label{stateprice}
&\E\left[\int_0^T c_s \Gamma_s ds\right]\nonumber\\
=&\E\left\{\int_0^T \left [C_s-h_0 e^{-\int_0^s \eta(u)du}-\psi(s) \int_0^s e^{-\int_u^s \eta(v)dv}C_u du\right ] \Gamma_s ds\right\}\nonumber\\
=&\E\left\{\int_0^T \left[\Gamma_s-\psi(s) \E_s\left(\int_s^T e^{\int_s^u \eta(v)dv}
\Gamma_u du\right )\right] C_sds\right\}\nonumber\\
&-\E\left[\int_0^T h_0 e^{-\int_0^s \eta(u)du} \Gamma_s
ds\right]\nonumber\\
=&\E\left[\int_0^T  H_s C_s ds\right]-h_0 \E\left[\int_0^T e^{-\int_0^s
\eta(u)du}\Gamma_s ds\right].
\end{align}
Fortunately, Fubini theorem could be applied to the generalized
function in the second equation. This helps to simplify the form of
the equation when we decompose the integral into the pre-retirement
and the post-retirement parts. \vskip 5pt The first term of
Eq.\eqref{stateprice} is estimated in \eqref{termone}, and  we
denote  the second term by
\begin{eqnarray}\nonumber
z\triangleq \E\left[\int_0^T e^{-\int_0^s \eta(u)du} \Gamma_s ds\right ],
\end{eqnarray}
whose exact value is derived in Appendix \ref{A.3}.
\vskip 5pt
Thus, we establish the subsidiary problem of the original problem, that is, maximizing the expected utility:
\begin{align}\label{sub pro}
\E\left[\int_0^{T} u_s(c_s) ds\right ],
\end{align}
where $u_s(c_s)=e^{-\rho s} u(c_s)$, $s\in [0,T]   $, with the restrictions:
\begin{align}\label{restriction}
\E\left [\int_0^T c_s \Gamma_s ds\right ]\leq A-h_0 z,
\end{align}
where
\begin{align*}
A = \E\left\{\int_0^T\left[ (1-k)H_sW_s1_{\{s\leq \tau\}}+ g(\tau)De^{\xi t}
H_s1_{\{s>\tau\}}\right ]ds\right\}.
\end{align*}
The exact value of $A$ is derived in Appendix \ref{A.3}. $A$ is the
discounted present value of the incomes of the wages (excluding
social insurance contribution) and the benefits from the social
insurance under the risk neutral probability $Q$. The risk neutral
probability $Q$ on $(\Omega, \mathcal{F})$ is defined by
\begin{align*}
\frac{dQ}{dP}{\big|}_{\mathcal{F}_T}=\exp \left\{ -\frac{1}{2}\theta^2 T-\theta B_T\right\}.
\end{align*}
Remarkably,  if the individual's initial wealth is $A$ and has no
subsequent incomes, the corresponding optimal asset allocation and
consumption policies will be equivalent to the original optimization
problem. \vskip 5pt In order to make the problem \eqref{sub pro}
well defined,  a prerequisite is added:
\begin{align} \label{prerequisite}
\E\left [\int_0^T -L \Gamma_s ds\right ]\leq A-h_0 z.
\end{align}
Thus, the incomes of the wages and the benefits are adequate to
support the minimal consumption level.

\begin{lemma}
The subsidiary problem \eqref{sub pro} obtains its maximum when
\begin{align}\label{crep}
c_s=c^*_s\triangleq\mathcal{Y}_s(\nu_\tau \Gamma_s), \ s\in [0, T],
\end{align}
where
\begin{align}
\mathcal{Y}_s(y)=\arg\max_{x\geq -L} \left\{u_s(x)-xy\right \},\ s\in [0, T].
\end{align}
And $\nu>0$ (the value can be $+\infty$) satisfies the following
equation:
\begin{align}\label{exist condition}
\E\left[\int_0^T \mathcal{Y}_s(\nu \Gamma_s) \Gamma_s ds\right]= A-h_0 z.
\end{align}
\end{lemma}
\begin{proof}
Define $f(x)\triangleq \E\left[\int_0^T \mathcal{Y}_s(x \Gamma_s)
\Gamma_s ds\right ]$. Then $f$ is a non-increasing continuous function
based on the property of $\mathcal{Y}_s$  and monotone convergence theorem. In addition, as
$\mathcal{Y}_s(0+)=\infty$ and $\mathcal{Y}_s(\infty)=-L$, we have
the estimations: $f(0+)=\infty$ and $f(\infty)=\E\left [\int_0^T -L
\Gamma_s ds\right ]$. According to the assumption that
$\E\left[\int_0^T -L \Gamma_s ds\right ]\leq A-h_0 z$, we know that
there exists a $\nu$ satisfying \eqref{exist condition}.

Furthermore, the maximal expected utility of \eqref{sub pro} can be
obtained by the following implicit form:
\begin{align*}
\E\left[\int_0^T V_s(\nu \Gamma_s)ds\right ]+\nu(A-h_0 z),
\end{align*}
where
\begin{align*}
V_s(y)\triangleq\max_{x\geq -L} \left\{u_s(x)-xy\right\}, \ s \in [0, T].
\end{align*}
Based on the definitions of $\mathcal{Y}$ and $V$,  it is easy to deduce that  $\E\left[\int_0^{T} u_s(c_s) ds\right]$\\
$\leq \E\left [\int_0^{T} (V_s(\nu \Gamma_s)+c_s\nu \Gamma_s) ds\right
]=\E\left [\int_0^T V_s(\nu \Gamma_s)ds\right ]+\nu(A-h_0 z)$. And
the equality holds when $c_s=\mathcal{Y}_s(\nu \Gamma_s)$.
\end{proof}
\vskip 5pt
We have established the necessary condition of the
existence of $c^*\triangleq\{c^*_t, \ 0\leq t\leq T\}$. We still
need to establish the corresponding  optimal consumption policy
 $C^*\triangleq \{C^*_t, \ 0\leq t\leq T\}$ and
optimal risky investment policy  $\pi^*\triangleq\{\pi^*_t , \ 0\leq
t\leq T\}$ to make the process $c^*$ attainable.
\vskip 5pt
 As $ h$ satisfies the ordinary
differential equation (abbr. ODE):
\begin{align*}
dh_t&=\left[\psi(t) C_t-\eta(t) h_t\right ]dt\\
&=\left[\psi(t)(c_t+h_t)-\eta(t)
h_t\right ]dt\nonumber\\
&=\left [\psi(t) c_t-(\eta(t)-\psi(t))h_t\right ]dt,\ 0\leq t\leq T,
\end{align*}
solving this ODE, we have, for $ t\in [0, T]$,
\begin{align}\label{hc}
h_t=e^{-\int_0^t(\eta(s)-\psi(s)) ds}h_0+\psi(t) \int_0^t e^{-\int_0^s (\eta(u)-\psi(u)) du} c_s ds.
\end{align}
By combining \eqref{crep} and \eqref{hc}, the explicit form of the
consumption level $C^*$ is given as follows:
\begin{align*} C^*_t=\mathcal{Y}_t(\nu
\Gamma_t)+e^{-\int_0^t(\eta(s)-\psi(s)) ds}h_0+\psi(t) \int_0^t
e^{-\int_0^s (\eta(u)-\psi(u)) du} \mathcal{Y}_s(\nu \Gamma_s)ds, \ \
t\leq T.
\end{align*}
However, the corresponding asset allocation policy  $ \pi^*  $ is
not easy to derive. We need to reformulate the wealth
process $X^*$ skillfully and construct the corresponding $\pi^*$ to
make the wealth process attainable. \begin{theorem} If Assumption
\eqref{prerequisite} is satisfied, the corresponding risky
investment policy $\pi^*$ is given by
\begin{align*}
\pi^*_t=\frac{1}{\sigma}(\frac{\varphi_t}{H_t}+\theta X^*_t), \ t\leq T,
\end{align*}
where the process $\varphi$ is defined in Eq. \eqref{mrt}.
\end{theorem}
\begin{proof} Based on SDE \eqref{wealth}, the optimal wealth
process $X^*$ has the boundary constraint that $X^*_T=0$, that is,
the optimum is attained when the individual consumes all the incomes
and has nothing at time $T$. Otherwise, he/she could consume more at
time $T$ to improve the overall utility. By combing \eqref{wealth}
with $X^*_T=0$, the wealth process $X^*$  is reformulated as
follows:
\begin{eqnarray}\label{O1}
\!\! \!\! X^*_t=\frac{1}{H_t}
\E_t\left\{\int_t^T\left[ C^*_sH_s-(1-k)H_sW_s1_{\{s\leq
\tau\}}-g(\tau)De^{\xi s} H_s1_{\{s>\tau\}}\right] ds\right\}
\end{eqnarray}
for $ t\in [0, T]$.
Let $M $ be an   $\mathbb{F}$-adapted, right continuous martingale
and satisfy:
\begin{eqnarray}\label{O2}
\!\! \!\! M_t=
 \E_t\left\{\int_0^T\!\!\left [ C^*_sH_s\!-\!(1-k)H_sW_s1_{\{s\leq
\tau\}}\!-\!g(\tau)De^{\xi s} H_s1_{\{s>\tau\}}\right ] ds\right \}, a.s.,
\end{eqnarray}
for $ t\in [0, T]$.
By martingale representation theorem, there exists a unique
square-integrable process $\varphi=\{\varphi_t, \ 0\leq t \leq T\}$
satisfying
\begin{align}\label{mrt}
M_t=\int_0^t \varphi_s dB_s, \ t\leq T.
\end{align}
Thus, using \eqref{O1}, \eqref{O2} and \eqref{wealth}, we have
\begin{eqnarray}\label{mnt}
M_t\!\!&=&\!\!\E_t\left\{\int_t^T \left [C^*_sH_s\!\!-\!\!(1-k)H_sW_s1_{\{s\leq \tau\}}\!\!-\!\!g(\tau)De^{\xi s} H_s1_{\{s>\tau\}} \right ]ds\right\}\nonumber\\
&&+\int_0^t \left [C^*_sH_s\!\!-(1-k)H_sW_s1_{\{s\leq \tau\}}-g(\tau)De^{\xi s} H_s1_{\{s>\tau\}}\right ] ds\nonumber\\
&=&\!\!X_t^*H_t\!-\!\int_0^t\left[ (1\!-\!k)H_sW_s1_{s\leq \tau}\!+\!g(\tau)De^{\xi s} H_s1_{\{s>\tau\}} \!-\!C^*_sH_s\right ]ds\nonumber\\
&=&\!\!\int_0^t \left[\sigma \pi^*_s H_s-\theta X^*_sH_s\right ] dB_s.
\end{eqnarray}
Comparing the diffusion terms of $M_t$ in \eqref{mrt} and
\eqref{mnt}, we have
\begin{align*}
\pi^*_t=\frac{1}{\sigma}(\frac{\varphi_t}{H_t}+\theta X^*_t).
\end{align*}
\end{proof}
The detailed calculation of $C^*$, $X^*$ and $\pi^*$ is derived in
Appendix \ref{A.4}.  Thus, the optimal consumption and asset
allocation policies are established with the given retirement time.
\vskip 5pt At last, we treat the optimal retirement time as a
pre-commitment policy and establish the optimum at time $0$. Because
the retirement time  is variable in the second step of optimization,
we rewrite the individual's value function as follows:
\begin{align*}
V(\tau)=\E\left [\int_0^{T} u_{s}(\mathcal{Y}_{s}(\nu_\tau
\Gamma_{\tau,s})) ds\right ],
\end{align*}
where $ \Gamma_{s}$, $A$, $z$ and $\nu$  are rewritten as
$\Gamma_{\tau,s}$, $A(\tau)$, $z(\tau)$ and $\nu_\tau$, and $\tau$
is the control variable.  $\nu_\tau$ is determined by
$\E\left[\int_0^T \mathcal{Y}_{s}(\nu_\tau \Gamma_{\tau,s})
\Gamma_{\tau,s} ds\right ]= A(\tau)-h_0 z(\tau)$. \vskip 5pt The
remained work is to solve $\arg\max_{\tau} \{V(\tau)\}$. The optimum
is attained when the value of $\tau$ is at the boundary points
$\tau_{min}$,  $\tau_{max}$  or the extreme point satisfying
$\frac{\partial V(\tau)}{\partial \tau}=0$. It is difficult to
derive the explicit form of $\tau$ and we solve it numerically.
Fortunately, we can study the influence of the exogenous parameters
on the optimal retirement time by an innovative analytical method.
The specific analysis is given in the next section. \vskip 15pt
\setcounter{equation}{0}
\section{\bf Theoretical and numerical implications}
In this section, we first introduce the certainty equivalence of the
actual consumption and the habitual consumption processes. Under the
reasonable parameter settings, we study the evolution of the optimal
consumption over time and try to explain the ``retirement
consumption puzzle". Then, we study the influence of the exogenous
parameters on the optimal consumption policies. Theoretically,
exogenous variables can be divided into two categories that affect
the amount of wealth and the habitual level, and they influence the
optimal consumption ultimately through these two intermediate
variables. The last, we use an innovative idea to study the
influence of the parameters on the optimal retirement time
analytically.
\vskip 10pt
\subsection{\bf Certainty equivalence of the actual and habitual consumption  processes}
Because the actual and habitual consumption processes are both
stochastic, we usually explore $\E[C^*_t]$ and $ \E[h^*_t]$ to study
the evolutions of the two processes over time. However, the
traditional explored method is not accurate, especially under the
S-shaped utility and the long time horizon. There is a very small
probability that the amounts of  the actual and
habitual consumption are magnificent because of excellent wealth
appreciation or radical investment policies. In this circumstance,
the magnificent consumption levels of some tracks will greatly raise
the level of  the expectations $\E[C^*_t]$ and $
\E[h^*_t]$ over the long time horizon. Thus, the traditional method
loses its representativeness. \vskip 5pt We naturally explore the
certainty equivalence of the actual and habitual consumption
processes to study the optimal policies. Under this perspective, the
process $\hat{C}$ and its corresponding habitual process $\hat{h}$
satisfying $u_t(\hat{C}_t-\hat{h}_t)=\E[u_t(C^*_t-h^*_t)]$ are
defined as the certainty equivalence of the consumption processes.
That is, when the individual's deterministic actual and habitual
consumption amounts are $\{\hat{C}_t, \ t\in [0, T]\} $ and
$\{\hat{h}_t, \ t\in [0,T]\}$, he/she obtains the same utility as
the one under  the stochastic consumption processes.
Particularly, $\hat{c}_t=\hat{C}_t-\hat{h}_t$ is the excess
consumption under certainty equivalence. It directly measures the
magnitude of the utility obtained at time $t$. \vskip 5pt In the
next subsection, we study the evolution of the optimal consumption
over time under the certainty equivalence perspective. \vskip 10pt
\subsection{\bf Parameter settings and the baseline model}
In this subsection, we first set up  the reasonable parameter
settings according  to the real data in the financial
market, the labor market and the social insurance market. \vskip 5pt
For the financial market, the risk-free interest rate is $r=0.02$,
and the expected return and the volatility of the risky asset are
$\mu=0.08$, $\sigma=0.4$, respectively. As such, the market price of
risk is $\theta=0.15$. For the rules of the labor market, we assume
that the parameters of the individual's wage income are
$\alpha=0.028$, $\beta=0.02$ and $W_0=10$. The individual begins to
work at the age of $25$ and his/her maximal survival age is $100$.
Besides, the individual can choose the actual retirement age within
$[50,80]$ and the statutory retirement age is $65$. Retiring at the
statutory retirement age, the individual can obtain the full
benefit.  As such, we have $t_0=0$, $T=75$, $\tau_{min}=25$,
$\tau_{max}=55$ and $\tau_{st}=40$. Furthermore, the parameters
related to the social insurance are given as follows. The
contribution rate is $k=0.2$ and the full benefit is
$D=6$ at time $0$. The elastic parameter in
the penalty variable is $\zeta=0.015$ and the growth rate of the
benefit is $\xi=0.018$. In general, if the individual retires at the
minimal retirement age, he/she will obtain nearly $70\%$ of the
 full benefit. \vskip 5pt For the settings of the
habitual consumption level, we assume that $\psi=\eta$ usually
holds. As such, the habitual level is the arithmetic average of the
past consumption levels.  And it is realistic to assume that
$\eta=\psi=0.05$. After retirement, both the level and the
sensitivity of the habitual consumption decline, and we have $l=0.6$
and $m=0.3$.  These two parameters are originally used to depict
the changes of the habitual consumption after retirement. Thus, we
choose the logical values and change their values to study the
impacts on the optimal consumption and retirement policies. Besides,
we assume that the initial habitual level is $h_0=6$ and the minimal
consumption constraint requires that $L=0.5$. For the parameters in
the S-shaped utility function, $\rho=0.04$, $\gamma=0.8$ and
$\kappa=2.25$ are realistic settings. \vskip 5pt Then, based on the
baseline model, we exhibit the certainty equivalence of the optimal
consumption $ \hat{C} $ and the habitual consumption $ \hat{h} $. First, we assume
that the retirement time is given at the statutory retirement time.
That is, $\tau=\tau_{st}=40$. Figure 1 shows the evolution of the
optimal consumption and the habitual consumption over time. We
observe a sharp decline in the actual and habitual consumption
levels at retirement. This is consistent with the empirical evidence
observed in the ``retirement consumption puzzle". Because the
incomes of the baseline model are relatively abundant, the actual
consumption is higher than the habitual consumption, and this leads
to the spiral rise of the two consumption levels. However, the
growth rates of the actual and the habitual consumption levels
decline after retirement, which is caused by the reduced sensitivity of
the habitual formulation. In addition, we find that the initial
consumption is higher than the wage income. In this circumstance,
the individual initiates some loans to improve the early consumption
level and obtains higher utility.
\begin{figure}[H]
\begin{center}
\includegraphics[scale=0.50]{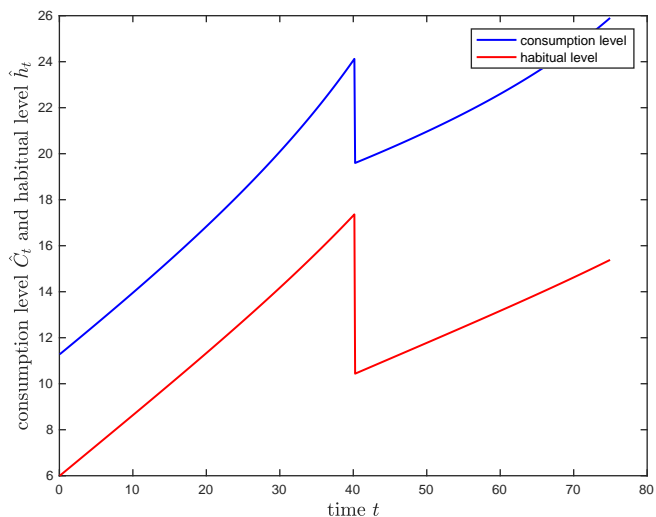}
\end{center}
\caption{The optimal consumption $\hat{C}$ and the habitual
consumption $\hat{h}$.}
\end{figure}
In Figure 2 and Figure 3, we select three typical retirement times
$\tau=35$, $40$, $ 45$  to study its influence on  the excess
consumption level $\hat{c}$ and the optimal consumption level
$\hat{C}$. \vskip 5pt Theoretically, the optimal excess
consumption level  $\hat{c}$ is balanced between two effects. The
first is the wealth effect. If the individual obtains more incomes,
he/she will naturally consume more. The second is the habitual
effect. Larger excess consumption rises the habitual consumption
level  even more and leads to more pressure on the
follow-up consumption. This prevents the individual from consuming
too much in the early time. If the habitual effect weakens, time
preference effect will be the dominance. That is, the early
consumption produces higher utility due to the human nature of
impatience. In this circumstance, the individual will consume more
at the moment, and vice verse. Thus, we observe a sharp rise in the
excess consumption level after retirement because of the reduced
sensitivity of the habitual  level and the weakened
habitual effect.
\begin{figure}[H]
\begin{center}
\includegraphics[scale=0.50]{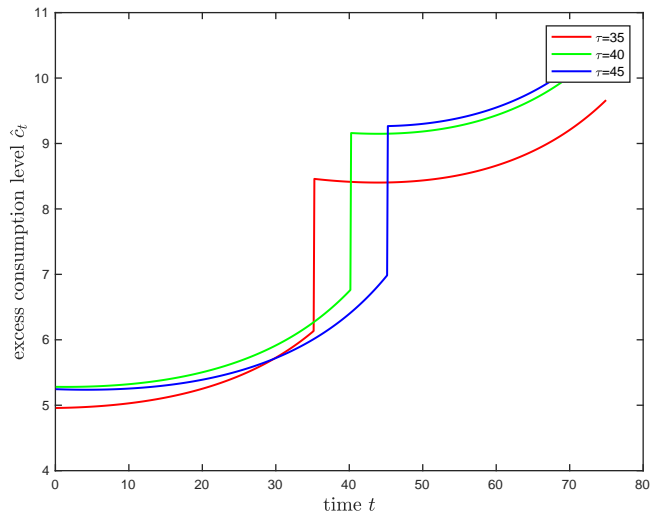}
\end{center}
\caption{The optimal excess consumption $\hat{c}$ with respect to
different retirement time $\tau$.}
\end{figure}
\begin{figure}[H]
\begin{center}
\includegraphics[scale=0.50]{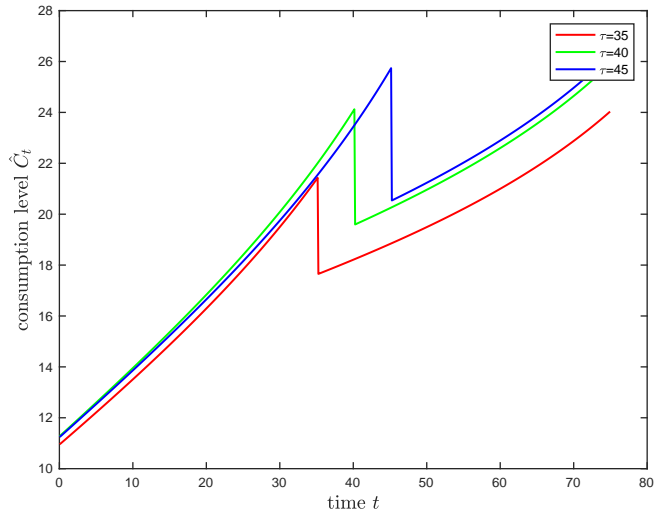}
\end{center}
\caption{The optimal consumption $\hat{C}$ with respect to different
retirement time $\tau$.}
\end{figure}
\vskip 5pt In Figure 2, the excess consumption levels are both low
before and after retirement in the $\tau=35$ case due to the lower
total incomes. Interestingly, the excess consumption is lower before
retirement and higher after retirement in the $\tau=45$ case than in
the $\tau=40$ case. Longer working time also means using the larger
habitual consumption benchmark for a longer time. Thus, the habitual
effect strengthens. The individual decreases the former consumption
and increases the latter consumption to balance the  two
effects. Furthermore, although the excess consumption rises a lot
at retirement, it cannot offset the sharp decline in the habitual
level. Thus, we observe a decline in the actual consumption
$\hat{C}$ at retirement in Figure 3. \vskip 5pt In Figure 4, we
study the optimal asset allocation policies over time
to ensure the integrity of the problem. As usual, we explore the
expectation of the amount  allocated to the risky asset. We observe
a downward curve due to the life cycle phenomenon. Thus, the special
variation assumption in habitual persistence after retirement
does not affect the trend of investment. \vskip 5pt Interestingly,
there is an unsmooth rise around the time $\tau_{min}=25$. Before
the time $\tau_{min}$, the individual's wage income is fluctuating.
However, after the time $\tau_{min}$, the individual receives static
wage and benefit incomes. In addition, the random sources of the
wage and the risky investment are the same. As such, the former
risky investment should be smaller to avoid bearing too much risk.
The latter risky investment can be larger as the risk bearing
ability is  improved. The result could also be drawn
from Eq. \eqref{pi}.
\begin{figure}[H]
\begin{center}
\includegraphics[scale=0.50]{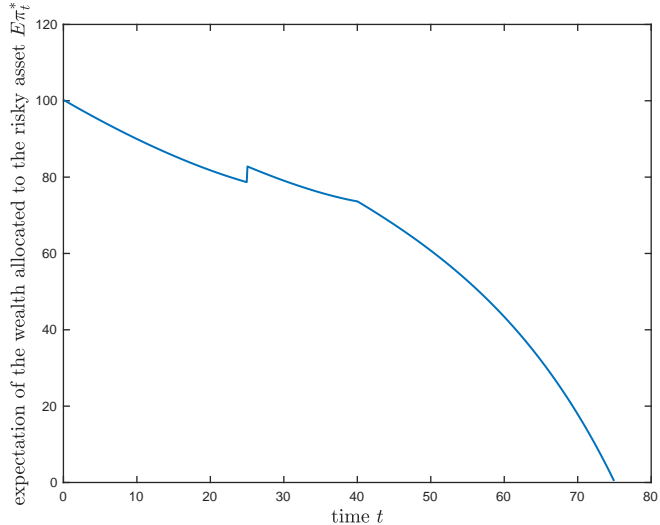}
\end{center}
\caption{The optimal risky investment amount $\E\pi^*_t$.}
\end{figure}
\vskip 10pt
\subsection{\bf Impacts of the parameters on the optimal consumption}
In this subsection, we study the impacts of the exogenous parameters
on the certainty equivalence of the optimal consumption $\hat{C}$
under the assumption that the individual retires at the statutory
retirement time. \vskip 5pt In Figure 5, we study the impacts of
total wealth $A$ (the discounted present value of the incomes) on
the optimal consumption $\hat{C}$. In fact,  the value
of $A$ is determined by a family of parameters, such as $W_0$, $D$,
$r$, $\mu$, $\sigma$, $k$, $\xi$ and $\zeta$, etc. The result
confirms the positive correlation between the total wealth and the
optimal consumption. When the total wealth $A$ is relatively large,
the trends of the optimal consumption are almost the same.
Interestingly, we observe that the percentage of decline in
consumption at retirement is negatively correlated with the total
wealth. This result is consistent with the empirical evidence that
the consumption drop is negatively correlated with income
replacement rate after retirement (cf. \cite{S2005}). However, when
$A$ is relatively small, the optimal consumption shows a declining
trend. In the case of $A=200$, which could be interpreted as $W_0=4$
and $D=2$, the total incomes are not enough to support the
consumption above the habitual level. Thus, the actual consumption
is lower than the habitual consumption. This leads to the spiral
decline of the both consumption levels and produces negative
utilities.
\begin{figure}[H]
\begin{center}
\includegraphics[scale=0.50]{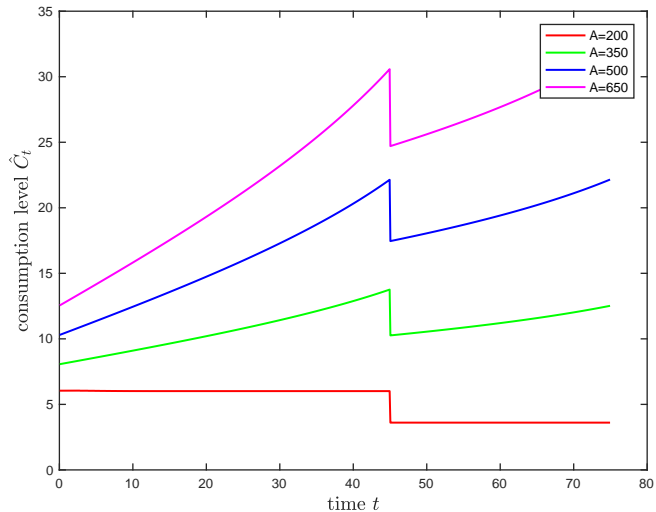}
\end{center}
\caption{The impacts of total wealth $A$ on the optimal consumption
$\hat{C}$.}
\end{figure}
In Figure 6, we study the impacts of the habitual parameters $\psi$
and $\eta$ on the optimal consumption $\hat{C}$. In order to make
the model well defined, we assume that $\psi=\eta$. As discussed in
Subsection 4.2, smaller $\psi$ and $\eta$ represent lower
sensitivity of the habitual level. Because the habitual effect
weakens, the time preference effect becomes the
dominance. As such, the individual increases former
consumption and reduces latter consumption. In the extreme case
$\psi=\eta=0$, we observe a completely declining trend of the
optimal consumption. In this circumstance, the habitual consumption
level is not affected by the former consumption. Thus, larger
consumption at the moment will produce higher utility according to
the time preference effect.
\begin{figure}[H]
\begin{center}
\includegraphics[scale=0.450]{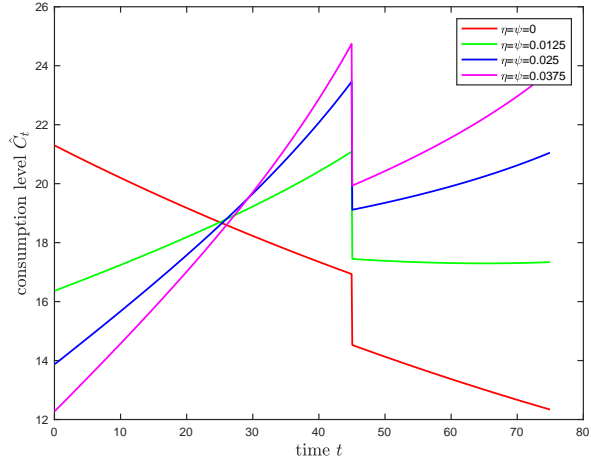}
\end{center}
\caption{The impacts of habitual parameters $\psi$ and $\eta$ on the
optimal consumption $\hat{C}$.}
\end{figure}
In Figure 7, we study the impacts of the shrinking habitual
proportion $l$ on the optimal consumption $\hat{C}$. $l$ measures the
magnitude of the decline in the level of habitual consumption after
retirement. Different $l$  leads to the redistribution of the
consumption levels before and after retirement. When $l$ is smaller,
the  lower consumption makes the individual feel the same
satisfactory as before retirement. Thus, he/she can consume more
before retirement. Although the consumption level after retirement
is reduced, the excess consumption is still considerable because of
the sharp decline of the habitual effect.
\vskip 5pt
\begin{figure}[H]
\begin{center}
\includegraphics[scale=0.450]{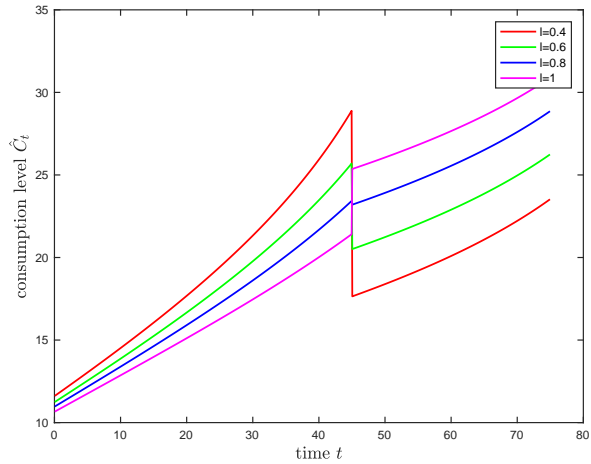}
\end{center}
\caption{The impacts of shrinking habitual proportion $l$ on the
optimal consumption $\hat{C}$.}
\end{figure}
\vskip 5pt In Figure 8, we study the  impacts of the shrinking
sensitivity parameter $m$ on the optimal consumption $\hat{C}$.
Because $m$ only influences  the sensitivity of the habitual  level
after retirement, its variation hardly has impacts on the
consumption levels before retirement. When $m$ is smaller, the
habitual effect weakens. In this circumstance, the individual
prefers to consume more in the early time because of the time
preference effect. Similar to the result in Figure 6, the optimal
consumption exhibits a declining trend after retirement in the
extreme case $m=0$.
\begin{figure}[H]
\begin{center}
\includegraphics[scale=0.50]{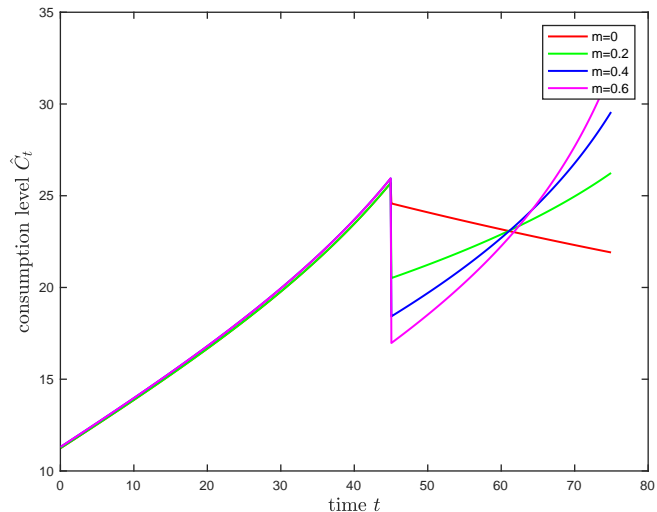}
\end{center}
\caption{The impacts of shrinking sensitivity parameter $m$ on the
optimal consumption $\hat{C}$.}
\end{figure}
In Figure 9, we study the impacts of the initial habitual
consumption level $h_0$ on the optimal consumption $\hat{C}$. The
result shows that the lower initial consumption increases the
admissible domain of the follow-up consumption. And the inertia of
the larger excess consumption in the former time leads to the higher
rise of the latter consumption. On the contrary, when $h_0$ is
large, the actual and habitual consumption levels in the latter time
are relatively small because of the small excess consumption in the
former time.
\begin{figure}[H]
\begin{center}
\includegraphics[scale=0.50]{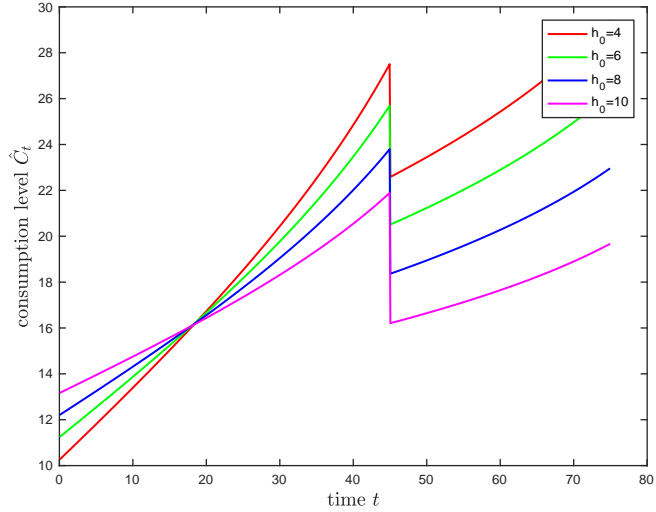}
\end{center}
\caption{The impacts of initial habitual consumption $h_0$ on the
optimal consumption $\hat{C}$.}
\end{figure}
\vskip 15pt
\subsection{\bf Optimal retirement time}
In this subsection, we treat the optimal retirement time as a
pre-commitment control variable and establish the optimal retirement
time to maximize the individual's value function. Then, we establish
an analytical method, which can be used to analyze the impacts of
the parameters on the optimal retirement time quantitatively. \vskip
5pt Before establishing the optimal retirement time, we first
analyze the two impacts of the retirement time on the overall
utility. The first is the wealth effect. The wealth effect can be
accurately estimated  by the function $A(\tau)$. In Figure 10, we
observe that $A(\tau)$ is an increasing function with respect to
$\tau$. From this perspective, the later an individual
retires, the more wealth and higher utility he/she can obtain. The
second is the habitual effect. The habitual level increases with the
extension of the working time according to the baseline model.
From this perspective, the later an individual retires,
he/she suffers from higher habitual consumption level and obtains
lower utility. In Figure 11, we modify the parameter settings to
preserve the validation of $A=500$, as such,  the influence of the
wealth effect is excluded. We observe a negative relationship
between the retirement time and the overall utility simply based on
the habitual effect. Overall, the optimal retirement
time is the balance between the two effects. \vskip 5pt
\begin{figure}[H]
\begin{center}
\includegraphics[scale=0.50]{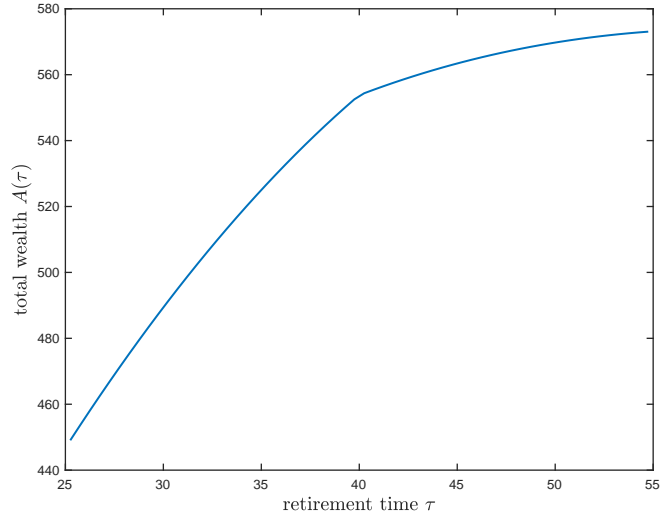}
\end{center}
\caption{The impacts of retirement time $\tau$ on the total wealth
$A(\tau)$.}
\end{figure}
\begin{figure}[H]
\begin{center}
\includegraphics[scale=0.50]{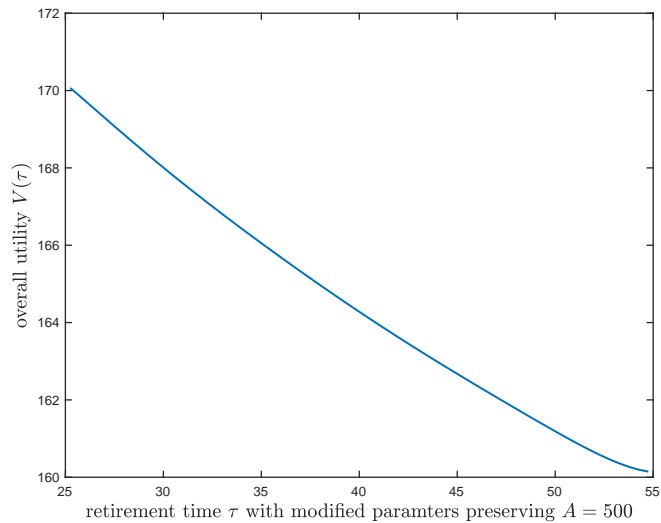}
\end{center}
\caption{The impacts of retirement time $\tau$ on overall utility
$V(\tau)$ based on habitual effect.}
\end{figure}
\vskip 5pt Next, we study the impacts of the retirement time on the
overall utility based on the two effects and establish the optimal
retirement time numerically. In Figure 12, the individual's utility
function first increases and then decreases with respect to the
extension of retirement time. In the former time, the wealth effect
dominates the habitual effect. And the habitual effect becomes the
dominance in the latter time. Interestingly, the optimal retirement
time is around $\tau_{st}=40$, which is exactly the
statutory retirement age of $65$, according to the
baseline model. \vskip 5pt
\begin{figure}[H]
\begin{center}
\includegraphics[scale=0.50]{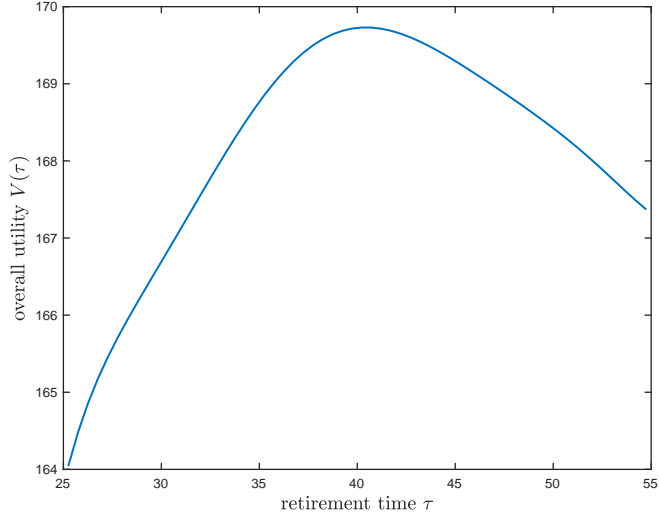}
\end{center}
\caption{The impacts of retirement time $\tau$ on overall utility
$V(\tau)$.}
\end{figure}
\vskip 5pt The last, the above discussions on the two effects are
helpful to study the impacts of the exogenous parameters on the
optimal retirement time. Besides, we establish a quantitative
analytical method as follows. For any exogenous parameter $y$, we
have the denotation that $\tau^*(y)\triangleq \arg\sup\limits_\tau\left\{
V(y,\tau)\right\}$. As such, $V_\tau(y,\tau^*(y))=0$ is valid.
Differentiating with respect to $y$ at both sides of this equality,
$\frac{d\tau^*(y)}{dy}=-\frac{V_{y\tau}(y,\tau^*(y))}{V_{\tau\tau}(y,\tau^*(y))}$,
where $V_{\tau\tau}(y,\tau^*(y))<0$. As such,
$\frac{d\tau^*(y)}{dy}$ has the same sign with
$V_{y\tau}\left((y,\tau^*(y)\right )$. Furthermore, the second-order
derivative can be approximated by
\begin{align}\label{secde} \frac{1}{\delta y \delta
\tau}\left[V(y+\delta y, \tau^*(y)\!+\!\delta \tau)\!-\!V(y,
\tau^*(y)\!+\!\delta \tau)\!-\!V(y\!+\!\delta y, \tau^*(y))\!+\!V(y,
\tau^*(y))\right].
\end{align}
Even if $V$ is not smooth, judging the sign  of \eqref{secde} is
still useful. The sign is determined by comparing the size of the
two values $V(y+\delta y, \tau^*(y)+\delta \tau)-V(y,
\tau^*(y)+\delta \tau)$ and $V(y+\delta y, \tau^*(y))-V(y,
\tau^*(y)))$, i.e.,  we can check whether the disturbance of $y$
will produce higher utility in the case of earlier retirement or
later retirement. Remarkably, we only exhibit the parameters of high
relevance and sensitivity in Table 1. Table 1 shows the relationship
between the exogenous parameters and the optimal retirement time.
The upward arrow indicates that the larger parameter leads to later
retirement, and vice verse. \vskip 5pt
\begin{table}[H]
\caption{The impacts of the parameters on optimal retirement time
$\tau^*$.}
\begin{tabular}{lllllllllll}
\hline
$W_0$ & $\alpha$  & $k$ & $D$ & $\xi$ & $\psi$ & $m$ & $l$ & $h_0$   \\
\hline
$\uparrow$ & $\uparrow$  & $\downarrow$  & $\downarrow$ & $\downarrow$ & $\downarrow$& $\downarrow$ &$\uparrow$ &$\downarrow$ \\
\hline
\end{tabular}
\end{table}
\vskip 5pt Echoing the discussions on the wealth effect and the
habitual effect, $W_0$, $\alpha$, $k$, $D$ and $\xi$  are the
parameters reflecting wealth effect, and $\psi$, $m$, $l$ and $ h_0$
 are the parameters reflecting habitual effect. Naturally, larger
initial level $W_0$ and  wage growth rate $\alpha$  will lead to
more wealth if the retirement is postponed. On the contrary, larger
$k$ leads to more contribution during the working period. And larger
$D$ and $\xi$ lead to more benefit during the retirement period.
Thus, these parameters make the individual retire early. \vskip 5pt
For the habitual effect, larger $\psi$ and  $m$ reflect higher
sensitivity of the habitual level. As such, earlier retirement is
required to prevent the habitual level from rising too high and
reducing overall utility. Besides, if the shrinking habitual
proportion $l$ is larger, the utility improvement by the
shrinking habitual level after retirement is less effective. Thus,
later retirement better balances the wealth effect and the habitual
effect in this circumstance. Naturally, early retirement is required
to control the habitual level when the initial habitual consumption
$h_0$ is large.
\vskip 25pt
 \setcounter{equation}{0}
\section{\bf Conclusion}
In this paper, we assume that the level and the sensitivity of the
habitual consumption both decline at the time of retirement. Under
the variation assumption in habitual persistence, we establish the
theoretical results consistent with the empirical evidence.  That
is, the optimal consumption experiences a sharp decline at
retirement. In fact, the optimal decisions are balanced between the
wealth effect and the habitual effect. Smaller growth inertia
(sensitivity) of the habitual level weakens the
habitual effect and
 leads to larger excess consumption after retirement. However,
this effect cannot offset the sharp decline in the habitual level.
Furthermore, the individual with higher habitual sensitivity retires
early to prevent the habitual level from rising too high.
Moreover, if the shrinking habitual parameter at
retirement is larger, delaying retirement is optimal
because of that the utility improvement by the shrinking
habitual level after retirement is less effective and the wealth
effect becomes the dominance.
\vskip 0.8cm {\bf Acknowledgements.}
The authors acknowledge the support from the National Natural
Science Foundation of China ( No.11871036, No.11471183).  The
authors also thank the members of the group of Mathematical Finance
and Actuarial Science at the Department of Mathematical Sciences,
Tsinghua University for their feedbacks and useful conversations.

\vskip 25pt
\appendix
\vskip 15pt
\section{Details of the calculations}
\renewcommand{\theequation}{\thesection .\arabic{equation}}
\vskip 10pt

\subsection{\bf The calculation of $\underline{C_t}$ and $\underline{X_t}$}\label{A.1}
Before deriving the exact form of $\underline{X_t}$, we claim that
the consumption level $C_t$ has the lower boundary
$\underline{C_t}$, which is the required minimal consumption:
\begin{align*}
\underline{C_t}=e^{-\int_0^t(\eta(s)-\psi(s)) ds}h_0-\psi(t) \int_0^t e^{-\int_0^s (\eta(u)-\psi(u)) du} Lds-L.
\end{align*}
In fact, $\underline{C_t}$ is obtained when the individual keeps the
consumption level with the gap $-L$ to the habitual consumption
level. As such, $\underline{X_t}$ can be written as
\begin{align*}
\underline{X_t}=\int_t^T \left [e^{-r(s-t)}\underline{C_s}-g(\tau)De^{\xi
s}e^{-r(s-t)}1_{\{s>\tau\}}\right ]ds.
\end{align*}
Remarkably, the condition $X_t\geq \underline{X_t}$ indicates that
the total wealth must be adequate to support the required minimal
consumption in the worst scenario that the wage income and the risky
investment return decay to $0$, i.e.,  Brownian motion $B   $
experiences a sharp
decline.\\
\vskip 10pt
\subsection{\bf The calculation of $\Gamma$}\label{A.2}

As $\frac{H_s}{H_t}=\exp \left\{ -r(s-t)-\frac{1}{2}\theta^2
(s-t)-\theta (B_s-B_t)\right\}$, $s\geq t$, is independent of
$\mathcal{F}_t$, we have \begin{align*}
\E_t\left\{\frac{H_s}{H_t}\right \}=\E\left\{\frac{H_s}{H_t}\right
\}=e^{-r(s-t)}.
\end{align*}
As such, for $t\in[0,T]$,
\begin{align*}
\Gamma_t&= H_t+\psi(t) \E_t\left\{\int_t^T e^{-\int_t^s (\eta(u)-\psi(u)) du}H_sds\right \}\\
&=H_t\left\{1+\psi(t) \E_t\left[\int_t^T e^{-\int_t^s (\eta(u)-\psi(u)) du}\frac{H_s}{H_t}ds\right]\right\}\\
&=H_t\left\{1+\psi(t) \int_t^T e^{-\int_t^s (\eta(u)-\psi(u)) du}e^{-r(s-t)}ds\right\}\\
&=H_t(1+F_t),
\end{align*}
where
\begin{eqnarray*}
F_t&=&\psi(t)\int_t^T e^{-\int_t^s (\eta(u)-\psi(u)) du}e^{-r(s-t)}ds\\
&=&\left\{
\begin{array}{lll}
\frac{l\psi e^{-(r+\eta-\psi)(\tau-t)}}{r+m\eta-m\psi}\left[1-e^{-(r+m\eta-m\psi)(T-\tau)}\right ]\\
+\frac{\psi}{r+\eta-\psi}\left[1-e^{-(r+\eta-\psi)(\tau-t)}\right ], \!\!\!\!\!\!\!\!\! &0\leq t<\tau,\\
\frac{m\psi}{r+m\eta-m\psi}\left[1-e^{-(r+m\eta-m\psi)(T-t)}\right], \!\! \! \!\!\!\!\!\!& \tau \leq t<T.
\end{array}
\right.
\end{eqnarray*}
\vskip 10pt
\subsection{\bf The calculation of $A$ and $z$}\label{A.3}

Similar to  Appendix \ref{A.2}, for $t\leq s\leq \tau_{min}$,
\begin{align*}
&\E_t\left(\frac{H_s}{H_t}\frac{W_s}{W_t}\right)=\E\left(\frac{H_s}{H_t}\frac{W_s}{W_t}\right)\\
=&\E\left[e^{(\alpha-\frac{\beta^2}{2})(s-t)+\beta(B_s-B_t)}e^{-r(s-t)-\frac{1}{2}\theta^2 (s-t)-\theta (B_s-B_t)}\right ]\\
=&e^{(\alpha-r+\theta\beta)(s-t)}.
\end{align*}
As such, we  derive the exact value of the restrictions in
\eqref{restriction} as follows: \begin{align}\label{A}
A&= \E\left\{\int_0^T\left[ (1-k)H_sW_s1_{\{s\leq \tau\}}+ g(\tau)De^{\xi s} H_s1_{\{s>\tau\}}\right ]ds\right\}\nonumber\\
&= \E_0\left\{\int_0^T \left[(1-k)H_sW_s1_{\{s\leq \tau\}}+ g(\tau)De^{\xi s} H_s1_{\{s>\tau\}}\right ]ds\right\}\nonumber\\
&= \int_0^T \left\{(1-k)\E_0[H_sW_s]1_{\{s\leq \tau\}}+ g(\tau)De^{\xi s} E_0[H_s]1_{\{s>\tau\}}\right\}ds\nonumber\\
&=\int_0^T \left[(1-k)W_0 (e^{(\alpha-r+\theta\beta)s}1_{\{s\leq \tau\cap s\leq \tau_{min}\}}\right.\nonumber\\
&\left.+e^{(\alpha-r+\theta\beta)\tau_{min}}e^{-r(s-\tau_{min})}1_{\{\tau_{min}<s\leq \tau  \}})+g(\tau)De^{\xi s}e^{-rs} 1_{\{s>\tau\}}\right]ds,
\end{align}
and the value of $z$ is
\begin{align}\label{z}
z&=\E\left[\int_0^T e^{-\int_0^s \eta(u)du} \Gamma_s ds\right ]\nonumber\\
&=\E\left [\int_0^T e^{-\int_0^s \eta(u)du} H_s(1+F_s)ds\right ]\nonumber\\
&=\int_0^T e^{-\int_0^s \eta(u)du} e^{-rs}(1+F_s)ds.
\end{align}
By combining \eqref{A} and \eqref{z}, the exact value of $A-h_0 z$ is
established.
\vskip 10pt
\subsection{\bf The calculation of $C^*$, $X^*$ and corresponding $\pi^*$}\label{A.4}

The function $\mathcal{Y}_s$ has the exact form that
\begin{align*}
\mathcal{Y}_s(y)=\left\{ \begin{array}{ll}
-L, &y\geq y_0(s),\\
(\frac{y}{e^{-\rho s}})^{-\frac{1}{\gamma}}, &0<y<y_0(s),
\end{array}\right.
\end{align*}
where $y_0(s)$ satisfies the following equation:
\begin{align*}
e^{-\rho s}\frac{(\frac{y_0(s)}{e^{-\rho
s}})^{-\frac{1-\gamma}{\gamma}}}{1-\gamma}-(-\kappa)e^{-\rho s}
\frac{L^{1-\gamma}}{1-\gamma}=y_0(s)\left[(\frac{y_0(s)}{e^{-\rho s}})^{-\frac{1}{\gamma}}+L\right].
\end{align*}
In fact, $y_0(s)$ is the gradient of the tangential line from $-L$ to $(\frac{y_0(s)}{e^{-\rho s}})^{-\frac{1}{\gamma}}$.
\vskip 5pt
Based on the  exact form of the function $\mathcal{Y}_s$, we have
\begin{align*}
\mathcal{Y}_s(\nu \Gamma_s)&=-L 1_{\{\nu \Gamma_s\geq y_0(s)\}}+\left(\frac{\nu \Gamma_s}{e^{-\rho s}}\right)^{-\frac{1}{\gamma}}1_{\{\nu \Gamma_s< y_0(s)\}}\\
&=-L 1_{\{ H_s\geq \frac{y_0(s)}{\nu (1+F_s)}\}}+H_s^{-\frac{1}{\gamma}} \left [\nu(1+F_s)e^{\rho s}\right]^{-\frac{1}{\gamma}}1_{\{ H_s< \frac{y_0(s)}{\nu (1+F_s)}\}}.
\end{align*}
As such,  we calculate the process  $X^*$ as follows:
\begin{align}\label{X}
X^*_t&=\frac{1}{H_t} \E_t\left\{\int_t^T \left [C^*_sH_s-(1-k)H_sW_s1_{\{s\leq \tau\}}-g(\tau)De^{\xi s} H_s1_{\{s>\tau\}}\right ] ds\right \}\nonumber\\
&=\frac{1}{H_t} \E_t\left[\int_t^T C^*_sH_s ds\right]-(1-k)\frac{1}{H_t} \E_t\left[\int_t^T H_sW_s1_{\{s\leq \tau\}}ds\right]\nonumber\\
&\quad -\frac{1}{H_t}\E_t\left[\int_t^T g(\tau)De^{\xi s} H_s1_{\{s>\tau\}} ds\right]\nonumber\\
&=\frac{1}{H_t} \E_t\left[\int_t^T C^*_sH_s ds\right]-(1-k) W_t \E_t\left[\int_t^T
\frac{H_sW_s}{H_t W_t}1_{\{s\leq \tau\}}ds\right]\nonumber\\
&\quad - \E_t\left[\int_t^T g(\tau)De^{\xi
s} \frac{H_s}{H_t}1_{\{s>\tau\}} ds\right].
\end{align}
The last two terms in \eqref{X} are
\begin{align*}
W_t \E_t\left[\int_t^T \frac{H_sW_s}{H_t W_t}1_{\{s\leq \tau\}}ds\right]=W_t O_t,
\end{align*}
where
\begin{align*}
O_t=&\int_t^T \big [e^{(\alpha-r+\theta \beta)(s-t)}1_{\{s\leq \tau\cup s\leq \tau_{min}\}}+e^{(\alpha-r+\theta\beta)(\tau_{min}-t)}e^{-r(s-\tau_{min})}1_{\{t<\tau_{min}<s\leq \tau  \}}\\
&+e^{-r(s-t)}1_{\{\tau_{min}\leq t<s\leq \tau  \}}\big]ds,
\end{align*}
and
\begin{align*}
\E_t\left[\int_t^T g(\tau)De^{\xi s} \frac{H_s}{H_t}1_{\{s>\tau\}}
ds\right ]=\int_t^T g(\tau)De^{\xi s} e^{-r(s-t)}1_{\{s>\tau\}} ds.
\end{align*}
The first term in \eqref{X} can be rewritten as follows:
\begin{align}\label{Xre}
&\frac{1}{H_t} \E_t\left[\int_t^T C^*_sH_s ds\right ]\nonumber\\
=&\frac{1}{H_t} \E_t\left\{\int_t^T \left[\mathcal{Y}_s(\nu \Gamma_s)+e^{-\int_0^s(\eta(u)-\psi(u)) du}h_0+\psi(s) \int_0^s e^{-\int_0^u (\eta(v)-\psi(v)) dv} \mathcal{Y}_u(\nu \Gamma_u)du\right ]H_s ds\right \}\nonumber\\
=&\int_t^T \left [e^{-\int_0^s(\eta(u)-\psi(u)) du}h_0+\psi(s) \int_0^t e^{-\int_0^u (\eta(v)-\psi(v)) dv} \mathcal{Y}_u(\nu \Gamma_u)du\right ]e^{-r(s-t)} ds\nonumber\\
&+\frac{1}{H_t} \E_t\left \{\int_t^T \left [\mathcal{Y}_s(\nu \Gamma_s)+\psi(s)
\int_t^s e^{-\int_0^u (\eta(v)-\psi(v)) dv} \mathcal{Y}_u(\nu
\Gamma_u)du\right ]H_s ds\right \}.
\end{align}
 The second equation holds as the integral can be divided into two parts:
 one   before time $t$ and the other   after time $t$, and the former part is $\mathcal{F}_t$-measurable.   We denote the latter part as $f_t(H_t)$:
\begin{align*}
f_t(H_t)&\triangleq\frac{1}{H_t} \E_t\left\{\int_t^T \left[\mathcal{Y}_s(\nu \Gamma_s)+\psi(s) \int_t^s e^{-\int_0^u (\eta(v)-\psi(v)) dv} \mathcal{Y}_u(\nu \Gamma_u)du\right ]H_s ds\right\}\\
&=\frac{1}{H_t} \E_t\left\{\int_t^T \left [\mathcal{Y}_s(\nu \Gamma_s)H_s+\psi(s) \int_t^s e^{-\int_0^u (\eta(v)-\psi(v)) dv} \mathcal{Y}_u(\nu \Gamma_u)\E[H_s|H_u] du\right ]ds\right\}\\
&=\frac{1}{H_t} \E_t\left\{\int_t^T \mathcal{Y}_s(\nu \Gamma_s)H_s\left[1+e^{-\int_0^s (\eta(v)-\psi(v)) dv} \int_s^T \psi(u) e^{-r(u-s)}du\right ]ds\right\}\\
&=\frac{1}{H_t} \E_t\left\{\int_t^T \mathcal{Y}_s(\nu \Gamma_s)H_sN_sds\right\}\\
&=\int_t^T \frac{1}{H_t} \E_t \left\{ \mathcal{Y}_s(\nu \Gamma_s)H_s \right\} N_sds\\
&=\int_t^T e^{-r(s-t)}N_s \left\{(e^{\rho s} y_0(s))^{-\frac{1}{\gamma}}\frac{\Phi'(d_{1,t,s}(\frac{y_0(s)}{\nu
(1+F_s)H_t}))}{\Phi'(d_{2,t,s}(\frac{y_0(s)}{\nu
(1+F_s)H_t}))}\left[1-\Phi(d_{2,t,s}(\frac{y_0(s)}{\nu
(1+F_s)H_t}))\right]\right.\\
&\quad \left.-\Phi(d_{1,t,s}(\frac{y_0(s)}{\nu (1+F_s)H_t}))L\right\} ds,
\end{align*}
where $\Phi(\cdot)$ is cumulative distribution function of the
standard normal distribution,

\begin{align*}
&N_s=1+e^{-\int_0^s (\eta(v)-\psi(v)) dv} \int_s^T \psi(u)
e^{-r(u-s)}du,\\
&d_{1,t,s}(x)=\frac{1}{-\theta\sqrt{s-t}}\left[\log(x)+(r+\frac{\theta^2}{2})(s-t)\right ],\nonumber\\
&d_{2,t,s}(x)=\frac{1}{-\theta\sqrt{s-t}}\left[\log(x)
+(r+\frac{\theta^2}{2})(s-t)+(1-\gamma)\theta\sqrt{s-t}\right].
\end{align*}
Thus, the wealth process can be rewritten as follows:
\begin{align*}
X^*_t=f_t(H_t)-(1-k)W_t O_t+\int_0^t \cdots ds.
\end{align*}
Based on martingale representation theorem and a huge amount of
detailed calculations,  using the volatility term of $X^*_t$, it
follows that \begin{align}\label{pi} \pi^*_t
=\frac{-f'_t(H_t)\theta H_t-(1-k)\beta W_t O_t 1_{\{t<\tau_{min}\}}}{\sigma}.
\end{align}
\end{document}